\documentclass[12pt]{article}
\usepackage[paper=a4paper,margin=1in]{geometry}
\usepackage{t1enc}      

\usepackage{amsmath}
\usepackage{amsfonts}
\usepackage{amssymb}
\usepackage{amsthm}
\usepackage{graphicx}
\usepackage{mathrsfs}  

\newcommand{\ua}{\underline a \,}
\newcommand{\ub}{\underline b \,}

\newcommand{\uA}{\underline A \,}

\newtheorem{theorem}{Theorem}[section]

\numberwithin{equation}{section}

\begin{document}
\bibliographystyle{unsrt}

\title{On the total mass of closed universes}
\author{L\'aszl\'o B Szabados \\
Wigner Research Centre for Physics, \\
H-1525 Budapest 114, P. O. Box 49, Hungary, \\
E-mail: lbszab@rmki.kfki.hu}

\maketitle

\begin{abstract}
The total mass, the Witten type gauge conditions and the spectral properties 
of the Sen--Witten and the 3-surface twistor operators in closed universes 
are investigated. It has been proven that a recently suggested expression 
${\tt M}$ for the total mass density of closed universes is vanishing if and 
only if the spacetime is flat with toroidal spatial topology; it coincides 
with the first eigenvalue of the Sen--Witten operator; and it is vanishing if 
and only if Witten's gauge condition admits a non-trivial solution. 

Here we generalize slightly the result above on the zero-mass configurations: 
${\tt M}=0$ if and only if the spacetime is \emph{holonomically trivial} with 
toroidal spatial topology. Also, we show that the multiplicity of the 
eigenvalues of the (square of the) Sen--Witten operator is even, and a 
potentially viable gauge condition is suggested. The monotonicity properties 
of ${\tt M}$ through the examples of closed Bianchi I and IX cosmological 
spacetimes are also discussed. A potential spectral characterization of these 
cosmological spacetimes, in terms of the spectrum of the Riemannian Dirac 
operator and the Sen--Witten and the 3-surface twistor operators, is also 
indicated.
\end{abstract}


\section{Introduction}
\label{sec:1}

In the present paper we discuss three, apparently independent issues: total 
masses and mass bounds, the spectral properties of certain differential 
operators, and gauge conditions on closed spacelike hypersurfaces. (A 
detailed discussion of these problems will be given in the following three 
subsections of the introduction, where we also formulate our specific 
questions.) However, it turned out that these questions are not independent, 
and our results provide common generalizations of previous particular ones. 
We review the key techniques and results (which have already appeared in 
\cite{Sz12}) in sections \ref{sec:2}--\ref{sec:5}, and, in section \ref{sec:3}, 
we generalize our previous result on the zero-mass spacetime configuration: 
the total mass density is zero if and only if the spacetime is 
\emph{holonomically trivial} (and not only locally flat) with toroidal spatial 
topology. Then, in section \ref{sec:6}, we illustrate these in the Bianchi I. 
and FRW spacetimes, and, in section \ref{sec:7}, we summarize the message of 
this contribution. 

Here we use the abstract index formalism, and our sign conventions are those 
of \cite{PR}. In particular, the signature of the spacetime metric is 
$(+,-,-,-)$, and the curvature tensor is defined by $-R^a{}_{bcd}X^bV^cW^d:=
V^c\nabla_c(W^d\nabla_dX^a)-W^c\nabla_c(V^d\nabla_dX^a)-[V,W]^c\nabla_cX^a$ for 
any vector fields $X^a$, $V^a$ and $W^a$. Thus, Einstein's equations take the 
form $G_{ab}:=R_{ab}-\frac{1}{2}Rg_{ab}=-\kappa T_{ab}$, where $\kappa:=8\pi G$ 
and $G$ is Newton's gravitational constant. 


\subsection{Total masses}
\label{sec:1.1}

It is well known that, as a consequence of the principle of equivalence, 
there is no well defined notion of gravitational energy-momentum 
\emph{density}. On the other hand, for spacetimes which are asymptotically 
flat at spatial or null infinity, there are well defined notions of total 
energy-momentum of the localized matter\,+\,gravity system. These are the 
Arnowitt--Deser--Misner (ADM) \cite{ADM62} and the Bondi--Sachs (BS) 
\cite{Bondi:etal62, Sachs62, Exton:etal69} energy-momenta, respectively. 
They are given by two-surface integrals on spheres at the spatial and (e.g. 
future) null infinity, rather than integrals of some local expression on 
spacelike hypersurfaces extending to the spatial and null infinity, 
respectively. The corresponding total masses are defined as their Lorentzian 
length. 

However, the form of these total energy-momentum expressions is rather 
different, and, in fact, completely different formalism and techniques are 
used to analyze the behaviour of the fields and the gravitation in the two 
asymptotic zones. Hence it is natural to search for a unified mathematical 
formulation of the two concepts of total energy-momentum. 

This unified form can be based on an appropriate superpotential $u(K)_{ab}$, 
which is a differential 2-form depending on some vector field $K^a$, such that 
in the so-called superpotential equation, 

\begin{equation}
\nabla_{[a}u\left(K\right)_{bc]}=J\left(K\right)_{abc}+\frac{1}{2}\kappa\, 
K^fT_{fe}\frac{1}{3!}\varepsilon^e{}_{abc}, \label{eq:1.1.1}
\end{equation}
the current 3-form $J(K)_{abc}$ is a homogeneous quadratic expression of the 
first derivative of the gravitational field variables. Here $\varepsilon
_{abcd}$ is the spacetime volume 4-form, and we used Einstein's equations. 
Then, for any smooth spacelike hypersurface $\Sigma$ with smooth boundary 
${\cal S}:=\partial\Sigma$, the `conserved quantity' generated by the 
spacetime vector field $K^a$ is defined by 

\begin{equation}
{\tt Q}\left[K\right]:=\frac{2}{\kappa}\oint_{\partial\Sigma}
u\left(K\right)_{ab}=\int_\Sigma\Bigl(\frac{2}{\kappa}J\left(K\right)_{abc}+
K^fT_{fe}\frac{1}{3!}\varepsilon^e{}_{abc}\Bigr). \label{eq:1.1.2}
\end{equation}
(If $\Sigma$ extends to the spatial or future null infinity in an 
asymptotically flat spacetime, then the integral on its boundary at infinity 
is defined in a limiting procedure.) Thus the key question is how to choose 
the generator vector field $K^a$ and the superpotential $u(K)_{ab}$ to recover 
the ADM/BS 4-momenta in the form ${\tt Q}[K]$? 

This question was addressed (among others) by Horowitz and Tod in 
\cite{Horowitz:Tod}. They showed that a particularly successful strategy is 
to use two-component spinors, and to choose the generator vector field to be 
$K^a=\lambda^A\bar\lambda^{A'}$ and the superpotential to be the Nester--Witten 
2-form 

\begin{equation}
u(\lambda)_{ab}:=\frac{\rm i}{2}\left(\bar\lambda_{A'}\nabla_{BB'}\lambda_A-
\bar\lambda_{B'}\nabla_{AA'}\lambda_B\right), \label{eq:1.1.3}
\end{equation}
where the spinor field $\lambda^A$ is still not specified. (The original form 
of the superpotential 2-form that Witten \cite{Wi} and Nester \cite{Ne} used 
was given in terms of Dirac spinors. The above form in terms of Weyl spinors 
was introduced by Horowitz and Tod \cite{Horowitz:Tod}.) To recover the ADM/BS 
energy-momenta the spacelike hypersurface $\Sigma$ should be chosen to be 
\emph{asymptotically flat}/\emph{asymptotically hyperboloidal}, and the spinor 
field $\lambda^A$ should tend to its own asymptotic value ${}_\infty\lambda^A$ 
in an appropriate order. In addition, the asymptotic value ${}_\infty\lambda^A$ 
should be \emph{constant} with respect to the two dimensional Sen type 
connection (see \cite{Sz94}), or should solve the \emph{asymptotic twistor 
equation} of Bramson \cite{Br}, respectively, on the metric 2-spheres at 
infinity. The geometric meaning of these boundary conditions is that these 
spinor fields are the \emph{spinor constituents of the asymptotic 
translations} at spatial and null infinity, respectively. The space of 
solutions of these equations form a two-complex dimensional spin space. If 
$\{\varepsilon^A_{\uA}\}$, ${\uA}=0,1$, is a normalized spin frame in these 
solution spaces, and ${}_\infty\lambda^{\uA}$ denotes the components of 
${}_\infty\lambda^A$ in this frame, then the components ${\tt P}_{\ua}$, 
${\ua}=0,1,2,3$, of the total energy-momenta can be recovered as 

\begin{equation}
{\tt P}_{\ua}\sigma^{\ua}_{\uA{\uA}'}\,{}_\infty\lambda^{\uA}\,{}_\infty\bar\lambda
^{{\uA}'}:=\frac{2}{\kappa}\lim_{r\rightarrow\infty}\oint_{{\cal S}_r}u\left(\lambda
\right)_{cd}, \label{eq:1.1.4}
\end{equation}
where $\sigma^{\ua}_{\uA{\uA}'}$ are the standard $SL(2,\mathbb{C})$ Pauli 
matrices. (Thus, while Latin indices are abstract tensor or spinor indices, 
the \emph{underlined} and (below) the \emph{boldface} Latin indices are 
concrete \emph{name} indices, taking numerical values.) 

Let $E^a_{\ua}:=\varepsilon^A_{\uA}\bar\varepsilon^{A'}_{{\uA}'}\sigma_{\ua}
^{{\uA}{\uA}'}$, the orthonormal vector basis determined by this spinor dyad. 
If on the asymptotically flat $\Sigma$ the basis $\{E^a_{\ua}\}$ is chosen to 
be adapted to the hypersurface $\Sigma$ (in the sense that $E^a_0$ is just the 
future pointing timelike normal of $\Sigma$), then the traditional ADM energy 
${\tt E}$ and linear momentum ${\tt P}_{\bf i}$, ${\bf i}=1,2,3$, can be 
recovered as the time and spacelike components, respectively, of ${\tt P}
_{\ua}$. If, in addition, the spinor field $\lambda^A$ is chosen to be 
normalized at infinity with respect to the timelike normal of $\Sigma$, i.e. 
${}_\infty t_{AA'}\,{}_\infty\lambda^A\,{}_\infty\bar\lambda^{A'}=1$, then the 
left hand side of (\ref{eq:1.1.4}) has the structure ${\tt E}+{\tt P}_{\bf i}
v^{\bf i}$, where $v^{\bf i}v^{\bf j}\delta_{\bf ij}=1$. We will need this form of 
the left hand side of equation (\ref{eq:1.1.4}). 

Similarly, if on the asymptotically hyperboloidal $\Sigma$ the spin frame 
$\{\varepsilon^A_{\uA}\}$ in the solution space of the asymptotic twistor 
equation is chosen such that the vector $E^a_0=\varepsilon^A_{\uA}\bar
\varepsilon^{A'}_{{\uA}'}\sigma_0^{{\uA}{\uA}'}$ is the BMS time translation of 
the future null infinity, then the BS energy and linear momentum are just 
the time and spacelike components, respectively, of ${\tt P}_{\ua}$. If, in 
addition, the spinor field is normalized such that $\sigma^0_{{\uA}{\uA}'}\,
{}_\infty\lambda^{\uA}\,{}_\infty\bar\lambda^{{\uA}'}=1$, then the left hand side 
of (\ref{eq:1.1.4}) has the structure ${\tt E}+{\tt P}_{\bf i}v^{\bf i}$. 

Then the positive energy theorems (see \cite{Wi, Ne, Reula, Horowitz:Tod, RT, 
Pa}) guarantee that the total energy-momenta, both the ADM and the BS, are 
future pointing and timelike with respect to the natural Lorentzian metric 
$\eta_{\ua\ub}:={\rm diag}(1,-1,-1,-1)$ coming from the symplectic scalar 
product of the solution spaces, provided the dominant energy condition is 
satisfied on the regular, asymptotically flat and asymptotically hyperboloidal 
$\Sigma$, respectively. Introducing the total mass ${\tt m}$ according to 
${\tt m}^2:=\eta_{\ua\ub}{\tt P}^{\ua}{\tt P}^{\ub}$, the positive energy 
theorems can be restated as ${\tt m}\geq0$. The rigidity part of these 
theorems guarantees that if the total mass is zero, ${\tt m}=0$, then the 
domain of dependence of $\Sigma$ is flat. 

The mass positivity results motivate the question whether we can find a 
\emph{strictly positive lower bound} for the total (ADM and BS) masses, i.e. 
that ${\tt m}^2\geq{\tt M}^2>0$ holds for some ${\tt M}$. Another question 
came from a recent result of B\"ackdahl and Valiente-Kroon \cite{BV}. They 
showed by explicit calculation that in vacuum, asymptotically flat spacetimes 
the ADM mass can be expressed as the $L_2$--norm on some spacelike 
hypersurface of the 3-surface twistor derivative of an appropriate spinor 
field: ${\tt m}_{ADM}={\rm const.}\Vert{\cal D}_{(AB}\lambda_{C)}\Vert_{L_2}^2$. 
Thus the question is whether we can have a similar expression for ${\tt m}
_{BS}$, too, and how this result can be generalized for the \emph{non-vacuum} 
case. 

Since the expected general form of the energy-momentum in General Relativity 
is a two-surface integral, it does not seem to be possible to associate any 
well defined notion of total energy-momentum, or at least total mass, with 
\emph{closed universes}. However, this does \emph{not} mean \emph{a priori} 
that a reasonable and useful notion of total mass cannot be associated with 
closed universes in some other way. This could perhaps be based on the idea 
that the total mass should be some \emph{positive definite measure} of the 
strength of the gravitational `field'.


\subsection{Spectral characterization of geometries}
\label{sec:1.2}

A potentially viable strategy to characterize Riemannian manifolds in an 
invariant way could be based on the study of the (structure of the) spectrum 
of elliptic (e.g. Laplace, Dirac, etc.) operators. The key idea is that e.g. 
the eigenvalues of elliptic operators encode (maybe in some highly 
non-explicit, but invariant way) certain properties of the geometry, and in 
the ideal case the whole geometry could be hoped to be characterized 
completely by the spectrum of a sufficiently large number of elliptic 
operators. Thus first it should be clarified e.g. how the eigenvalues reflect 
the properties of the geometry. 

The first who obtained such a link between the eigenvalues and certain 
properties of the geometry was probably Lichnerowicz \cite{Lich}. He showed 
that on closed $m$ dimensional spin manifolds $M$ with non-negative scalar 
curvature $R$ for the 1st eigenvalue of the Dirac operator, ${\rm i}\gamma
^\alpha_e{}_\beta D^e\Psi^\beta=\alpha_1\Psi^\alpha$, one has $\alpha^2_1\geq
\frac{1}{4}\inf\{ R(p)\vert p\in M\}$. This lower bound is, however, not 
sharp. Lichnerowicz's bound was increased by Friedrich \cite{Fr} by giving 
the \emph{sharp} lower bound: $\alpha^2_1\geq\frac{m}{4(m-1)}\inf\{ R(p)\vert 
p\in M\}$. This bound is, in fact, saturated by the metric spheres. Later, 
several other sharp bounds were derived under various geometrical conditions. 

From the point of view of General Relativity it would be desirable to extend 
the above results from Riemannian manifolds to \emph{initial data sets}. In 
fact, Hijazi and Zhang \cite{HZ} derived a sharp lower bound for the 1st 
eigenvalue of the Sen--Witten operator on closed hypersurfaces in Lorentzian 
geometries. In terms of the standard notions in GR their bound is 

\begin{equation}
\alpha^2_1\geq\frac{3}{4}\kappa\,\inf_{l^a}\,\frac{\int_\Sigma t^aT_{ab}l^b
{\rm d}\Sigma}{\int_\Sigma t_cl^c{\rm d}\Sigma}, \label{eq:1.2.1}
\end{equation}
where the infimum is taken on the set of all the future pointing null vector 
fields $l^a$ on the hypersurface. Thus the bound is the infimum of an average 
of the \emph{total matter energy} in $\Sigma$. (The bound of Hijazi and Zhang 
was rediscovered independently in \cite{Sz07}, and the form (\ref{eq:1.2.1}) 
of their bound is taken from \cite{Sz07}.) 

Nevertheless, the lower bound (\ref{eq:1.2.1}) is zero in vacuum, giving no 
restriction on the eigenvalues. This motivates the question whether we can 
find an even greater, and hence sharp, lower bound which is not trivial even 
in vacuum. A more ambitious claim is to derive an \emph{expression for the 
first eigenvalue} itself, rather to have only a lower bound for it. 


\subsection{Gauge conditions}
\label{sec:1.3}

In various specific problems of General Relativity (e.g. in the energy 
positivity proofs, evolution problems, numerical calculations, etc) it is 
desirable to reduce the huge gauge freedom of the theory. These conditions 
are used to single out some `preferred' frame of reference, which are built 
from a special spinor field (see e.g. \cite{joerg}). Such conditions are, 
for example, the Witten \cite{Wi}, the Parker \cite{Pa} and the Nester 
\cite{Ne2} gauge conditions; and the so-called approximate twistor equation 
of B\"ackdahl and Valiente-Kroon \cite{BV} can also be interpreted as a gauge 
condition. In the present contribution we discuss only the gauge conditions 
of Witten and of B\"ackdahl and Valiente-Kroon. 

Witten's gauge condition is simply the differential equation ${\cal D}_{A'A}
\lambda^A=0$ for the spinor field on $\Sigma$ which satisfies some boundary 
condition, where ${\cal D}_a:=P^b_a\nabla_b$ is the projection to $\Sigma$ of 
the spacetime Levi-Civita derivative operator $\nabla_a$, known as the Sen 
connection. Similarly, the approximate twistor equation on $\Sigma$ is 
${\cal T}^*{\cal T}(\lambda)=0$, which is the composition of the 3-surface 
twistor operator ${\cal T}$ of Tod \cite{Tod84} and its formal adjoint 
${\cal T}^*$. The former is defined to be the totally symmetric part of the 
derivative, i.e. ${\cal T}:\lambda_A\mapsto{\cal D}_{(AB}\lambda_{C)}$, where 
${\cal D}_{AB}$ is the unitary spinor form of ${\cal D}_{AA'}$ (see e.g. 
\cite{Reula}). The formal adjoint ${\cal T}^*$ is defined with respect to 
the $L_2$ scalar product, and we give its explicit form by (\ref{eq:5.1}) in 
section \ref{sec:5}. It is known that both the Witten and the approximate 
twistor equations admit non-trivial solutions on asymptotically flat 
hypersurfaces. (For the proofs see e.g. \cite{RT} and \cite{BV}, respectively. 
The existence of solutions of the Witten equation on asymptotically 
hyperboloidal hypersurfaces was also demonstrated in \cite{RT}.) 

Nevertheless, as far as we know, it has not been clarified whether these 
gauge conditions admit non-trivial solutions on \emph{closed} spacelike 
hypersurfaces. Hence it is natural to ask whether these gauge conditions can 
be imposed in closed universes, and if not, then how they can be modified to 
obtain an appropriate one.


\section{The norm identity}
\label{sec:2}

In our investigations the key geometric ingredient is the Reula--Tod (or 
$SL(2,\mathbb{C})$ spinor) form \cite{RT} of the Sen--Witten identity: 

\begin{eqnarray}
D_a\!\!\!\!&\Bigl(\!\!\!\!&t^{A'B}\bar\lambda^{B'}{\cal D}_{BB'}\lambda^A-
 \bar\lambda^{A'}t^{AB'}{\cal D}_{B'B}\lambda^B\Bigr)+2t^{AA'}\bigl({\cal D}_{AB'}
 \bar\lambda^{B'}\bigr)\bigl({\cal D}_{A'B}\lambda^B\bigr)\nonumber \\
\!\!\!\!&=\!\!\!\!&-t_{AA'}h^{ef}\bigl({\cal D}_e\lambda^A\bigr)\bigl(
 {\cal D}_f\bar\lambda^{A'}\bigr)+\frac{1}{2}\kappa t^aT_{ab}\lambda^B\bar
 \lambda^{B'}, \label{eq:2.1}
\end{eqnarray}
where $t^a$ is the future pointing unit timelike normal of $\Sigma$, $P^a_b=
\delta^a_b-t^at_b$ the projection to $\Sigma$, and $h_{ab}:=P^c_aP^d_bg_{cd}$ is 
the induced \emph{negative definite} metric on $\Sigma$. This equation is just 
the Hodge dual of the pull back to $\Sigma$ of the superpotential equation 
(\ref{eq:1.1.1}), in which the superpotential is the Nester--Witten 2-form, 
and Einstein's equation is used. The total divergence on the left is just the 
Hodge dual of the pull back to $\Sigma$ of the exterior derivative of 
$u(\lambda)_{ab}$. 

The key observation is that the (algebraically) irreducible decomposition 
of the unitary spinor form \cite{Reula} of the ${\cal D}_e$--derivative of 
the spinor field into its totally symmetric part and the traces, 

\begin{equation}
{\cal D}_{EF}\lambda_A={\cal D}_{(EF}\lambda_{A)}+\frac{2\sqrt{2}}{3}
t_F{}^{E'}P^{CC'}_{EE'}\varepsilon_{CA}{\cal D}_{C'D}\lambda^D, 
\label{eq:2.2}
\end{equation}
is $t_{AA'}$-orthogonal, and hence it is an $L_2$--orthogonal decomposition 
also. Here the totally symmetric part of the derivative defines the 3-surface 
twistor operator, while the second term is proportional to the action of the 
Sen--Witten operator (i.e. the Dirac operator built from the Sen connection 
on $\Sigma$) on the spinor field. 

Substituting this decomposition into the Sen--Witten identity and taking its 
integral on $\Sigma$ we obtain 

\begin{equation}
{\tt Q}[\lambda\bar\lambda]+\frac{4\sqrt{2}}{3\kappa}\Vert{\cal D}_{A'A}
\lambda^A\Vert^2_{L_2}=\frac{\sqrt{2}}{\kappa}\Vert{\cal D}_{(AB}\lambda_{C)}
\Vert^2_{L_2}+\int_\Sigma t^aT_{ab}\lambda^B\bar\lambda^{B'}{\rm d}\Sigma. 
\label{eq:2.3}
\end{equation}
Here ${\tt Q}[\lambda\bar\lambda]$ is the integral of $\frac{2}{\kappa}
u(\lambda)_{ab}$ on the boundary $\partial\Sigma$, and the $L_2$--norm of the 
spinor field is defined to be the integral on $\Sigma$ of the pointwise 
Hermitian scalar product $\sqrt{2}t_{AA'}\lambda^A\bar\lambda^{A'}$. This 
equation will play a key role in what follows and we call it the \emph{basic 
norm identity}. Clearly, if $\Sigma$ is asymptotically flat or asymptotically 
hyperboloidal, then the existence of the $L_2$ norms and the integral of the 
energy-momentum tensor require appropriate fall-off properties both for the 
geometry and the energy-momentum tensor, and also for the spinor fields.


\section{Total masses and mass bounds}
\label{sec:3}

\subsection{The asymptotically flat/asymptotically hyperboloidal cases}
\label{sec:3.1}

By the expression (\ref{eq:1.1.4}) of the ADM/BS energy-momenta, the basic 
norm identity (\ref{eq:2.3}) and Witten's gauge condition we obtain the 
following simple expression for the total energy-momenta: 

\begin{equation}
{\tt P}_{\ua}\sigma^{\ua}_{\uA{\uA}'}\,{}_\infty\lambda^{\uA}\,{}_\infty\bar\lambda
^{{\uA}'}=\frac{\sqrt{2}}{\kappa}\Vert{\cal D}_{(AB}\lambda_{C)}\Vert^2_{L_2}+
\int_\Sigma t^aT_{ab}\lambda^B\bar\lambda^{B'}{\rm d}\Sigma. \label{eq:3.1.1}
\end{equation}
This expression is an extension of the result of B\"ackdahl and Valiente-Kroon 
for the ADM energy in vacuum, mentioned in subsection \ref{sec:1.1}: It gives 
\emph{both} the ADM and BS total energy-momenta, even in the non-vacuum case, 
in terms of the $L_2$--norm of the 3-surface twistor derivative of an 
appropriate spinor field and the integral of the energy-momentum tensor of the 
matter fields. 

The right hand side of (\ref{eq:3.1.1}) motivates the introduction of the 
following quantity:

\begin{equation}
{\tt M}:=\inf\Bigl\{\frac{\sqrt{2}}{\kappa}\Vert{\cal D}_{(AB}\lambda_{C)}
\Vert^2_{L_2}+\int_\Sigma t^aT_{ab}\lambda^B\bar\lambda^{B'}{\rm d}\Sigma
\Bigr\}. \label{eq:3.1.2}
\end{equation}
Here the infimum is taken on the set of the spinor fields satisfying the 
boundary and normalization conditions that we had in the asymptotically 
flat/asymptotically hyperboloidal cases, respectively. 

Let us write the components of the energy-momentum as ${\tt P}^{\ua}=({\tt E},
{\tt P}^{\bf i})$. Then, by the positive energy theorem ${\tt E}\geq\vert{\tt 
P}^{\bf i}\vert$, for the total mass we obtain that ${\tt m}^2:={\tt P}^{\ua}
{\tt P}^{\ub}\eta_{\ua\ub}={\tt E}^2-\vert{\tt P}^{\bf i}\vert^2=({\tt E}-
\vert{\tt P}^{\bf i}\vert)({\tt E}+\vert{\tt P}^{\bf i}\vert)\geq({\tt E}-
\vert{\tt P}^{\bf i}\vert)^2$, where $\vert{\tt P}^{\bf i}\vert$ denotes the 
magnitude of the linear momentum ${\tt P}^{\bf i}$. However, since ${\tt E}-
\vert{\tt P}^{\bf i}\vert=\inf\{\,{\tt E}+{\tt P}_{\bf i}v^{\bf i}\,\vert\, 
v^{\bf i}v^{\bf j}\delta_{\bf ij}=1\,\}={\tt M}$, the expression (\ref{eq:3.1.2}) 
provides a lower bound for ${\tt m}$. (Here the infimum is taken on the set of 
the unit vectors $v^{\bf i}$.) By the rigidity part of the positive energy 
theorems (i.e. if ${\tt E}=\vert{\tt P}^{\bf i}\vert$ then the spacetime is 
flat), the lower bound ${\tt M}$ for the ADM and BS masses is \emph{strictly 
positive}, unless the spacetime is flat.

\subsection{The closed case}
\label{sec:3.2}

If $\Sigma$ is closed, i.e. compact with no boundary, then no total 
energy-momentum (and hence mass) can be introduced in the form of a two-surface 
integral. However, the quantity ${\tt M}$ can still be introduced by 
(\ref{eq:3.1.2}) with an appropriately chosen set of spinor fields on which 
the infimum is taken. Clearly, since $\Sigma$ is closed, we do not have any 
boundary condition for the spinor fields, but we should impose a normalization 
condition. The most natural such condition seems to be $\Vert\lambda^A\Vert^2
_{L_2}=1$, which was used in \cite{Sz12}. However, the condition $1=\int_\Sigma 
t_{AA'}\lambda^A\bar\lambda^{A'}{\rm d}\Sigma=\frac{1}{\sqrt{2}}\Vert\lambda^A
\Vert^2_{L_2}$ is even more natural, since it is just the integral of the 
pointwise norm $t_{AA'}\lambda^A\bar\lambda^{A'}$ whose asymptotic form, 
${}_\infty t_{AA'}\,{}_\infty\lambda^A\,{}_\infty\bar\lambda^{A'}=1$, was used to 
normalize the spinor field e.g. at spatial infinity in the asymptotically 
flat case. Another normalization condition could be $\Vert\lambda^A\Vert^2_{L_2}
=\sqrt{2}{\rm vol}(\Sigma)$, i.e. the `average' of the pointwise norm $t_{AA'}
\lambda^A\bar\lambda^{A'}$ on $\Sigma$ would be required to be 1. The physical 
dimension of ${\tt M}$ in the two cases is mass-density and mass, 
respectively. However, although the second choice appears to yield a 
dimensionally correct mass expression, for later convenience we choose the 
first normalization condition. Note that with this normalization ${\tt M}$ 
is $\sqrt{2}$-times of the ${\tt M}$ introduced in \cite{Sz12}. 

Since $\Sigma$ does not have a boundary, the basic norm identity 
(\ref{eq:2.3}) for any spinor field and the definition of ${\tt M}$ yield 
that 

\begin{equation}
\frac{4\sqrt{2}}{3\kappa}\Vert{\cal D}_{A'A}\lambda^A\Vert^2_{L_2}=\Bigl\{
\frac{\sqrt{2}}{\kappa}\Vert{\cal D}_{(AB}\lambda_{C)}\Vert^2_{L_2}+\int
_\Sigma t^aT_{ab}\lambda^B\bar\lambda^{B'}{\rm d}\Sigma\Bigr\}\geq\frac{1}
{\sqrt{2}}{\tt M}\Vert\lambda^A\Vert^2_{L_2}. \label{eq:3.2.1}
\end{equation}
Since ${\tt M}$ was defined as the infimum of an expression on a set of 
certain smooth spinor fields, it is not \emph{a priori} obvious that there 
is a smooth spinor field which saturates the inequality on the right. 
Nevertheless, one can in fact show that \emph{such a spinor field does 
exist} \cite{Sz12}. We will call such a spinor field a \emph{minimizer} 
spinor field. Thus, if $\lambda^A$ is such a minimizer spinor field and 
$\langle\,\cdot,\cdot\,\rangle$ denotes the $L_2$ scalar product, then, for 
\emph{this} spinor field, by (\ref{eq:3.2.1}) we have that 

\begin{equation}
\langle\, 2{\cal D}^{AA'}{\cal D}_{A'B}\lambda^B-\frac{3}{4}\kappa\,{\tt M}
\lambda^A\, ,\,\lambda^A\rangle=2\Vert{\cal D}_{A'A}\lambda^A\Vert^2_{L_2}-
\frac{3}{4}\kappa\,{\tt M}\Vert\lambda^A\Vert^2_{L_2}=0. \label{eq:3.2.2}
\end{equation}
This implies that either the minimizer spinor field is necessarily 
$L_2$--orthogonal to the spinor field $2{\cal D}^{AA'}{\cal D}_{A'B}\lambda^B-
\frac{3}{4}\kappa\,{\tt M}\lambda^A$, or that $\frac{3}{4}\kappa\,{\tt M}$ is 
an \emph{eigenvalue} and the minimizer spinor field is a corresponding 
\emph{eigenspinor} of the operator $2{\cal D}^{AA'}{\cal D}_{A'B}$. We will see 
in the next section that this is, indeed, the case, and $\frac{3}{4}\kappa\,
{\tt M}$ is its \emph{smallest} eigenvalue. 

The (geometrical and physical) significance of ${\tt M}$ is shown by the result 
\cite{Sz12} that the vanishing of ${\tt M}$ is equivalent to the flatness of 
the spacetime with toroidal spatial topology. However, if we allow to have 
locally flat, but holonomically non-trivial spacetime configurations, then 
in such domains the constant spinor fields are not necessarily continuous 
everywhere, and hence our previous theorem should be modified. This 
possibility motivated the following generalization of our previous result:  

\begin{theorem}
Let the matter fields satisfy the dominant energy condition. Then ${\tt M}=0$ 
for some (and hence for any) $\Sigma$ if and only if the spacetime is 
holonomically trivial and the topology of $\Sigma$ is torus: $\Sigma\approx 
S^1\times S^1\times S^1$. \label{th:3.1}
\end{theorem}
\begin{proof}
Since the detailed proof (with a different line of argument) of the original 
statement is given in \cite{Sz12}, here we only summarize its key points and 
concentrate on the difference between the present and the original statements. 
Clearly, in holonomically trivial spacetime with $\mathbb{R}\times S^1\times 
S^1\times S^1$ global topology there are globally defined constant spinor 
fields, which satisfy the 3-surface twistor equation. Therefore, ${\tt M}=0$. 

Conversely, suppose that ${\tt M}=0$. Then the minimizer spinor field 
satisfies ${\cal D}_{(AB}\lambda_{C)}$ $=0$, and hence by the basic norm 
identity ${\cal D}_e\lambda^A=0$ follows. This implies that $Z_a:=P^b_a\lambda
_B\bar\lambda_{B'}$ is surface-orthogonal, i.e. $Z_a=D_au$ for some (locally 
defined) function $u$ on $\Sigma$. Using the Gauss equation for hypersurfaces 
one can show that the level sets ${\cal S}_u:=\{\,u={\rm const}\,\}$ are 
\emph{locally flat} 2-surfaces in $\Sigma$. 

Next, one can show that the level sets ${\cal S}_u$ are \emph{closed 
surfaces}, which by the Gauss--Bonnet theorem are necessarily \emph{two-tori}. 
$\Sigma$ is globally foliated by these tori, and hence it is a 3-torus. 

Finally, we use the field equations. In an appropriate spacetime coordinate 
system adapted to the geometry they reduce to a single Poisson equation with 
a source term on the level sets ${\cal S}_u$. This source is, however, 
non-negative by the dominant energy condition. To have a non-trivial solution 
to this equation both the source term and the not a priori zero components of 
the curvature have to be zero, i.e. the spacetime is flat. Let $\gamma$ be any 
smooth homotopically non-trivial closed curve. Then the constant spinor field 
$\lambda^A$ is clearly parallelly propagated along $\gamma$. However, if the 
holonomy $H_\gamma$ were not trivial, then $\lambda^A$ could not be continuous 
everywhere along $\gamma$. Thus the whole holonomy group must be trivial. 
\end{proof}

By definition, ${\tt M}$ is non-negative, and hence the content of this 
theorem is analogous to the \emph{rigidity part} of the positive energy 
theorems for the ADM/BS masses. Thus ${\tt M}$ is a \emph{positive definite 
measure} of the strength of the gravitational field. Since the \emph{physical 
dimension} of ${\tt M}{\rm vol}(\Sigma)$ is mass, moreover it is given by 
precisely the formula that we had in the previous subsection for the positive 
lower bound for the ADM/BS masses, it seems plausible to interpret ${\tt M}$ 
as the \emph{total mass density of closed universes} at the instant represented 
by the hypersurface $\Sigma$.


\section{The eigenvalue problem}
\label{sec:4}

Since in four dimensions the spinors are the four-component Dirac spinors 
(see e.g. the appendix of \cite{PR}), and the Sen connection from which the 
Sen--Witten operator is constructed is defined on a vector bundle over 
$\Sigma$ whose fibers are four dimensional Lorentzian vector spaces, the 
eigenvalue problem should be formulated in terms of Dirac spinors. (For 
the discussion of the difficulties with other approaches, see \cite{Sz12}.) 
Thus, the eigenvalue problem for the Sen--Witten operator is defined by 
${\rm i}\gamma^\alpha_{e\beta}{\cal D}^e\Psi^\beta=\alpha\Psi^\alpha$, where 
the Greek indices are abstract indices referring to the space of the Dirac 
spinors, and $\gamma^\alpha_{e\beta}$ are Dirac's `$\gamma$-matrices'. Recalling 
that a Dirac spinor $\Psi^\alpha$ is a pair $(\lambda^A,\bar\mu^{A'})$ of Weyl 
spinors, with the explicit form of $\gamma^\alpha_{e\beta}$ given in \cite{PR} 
the eigenvalue problem is equivalent to the pair 

\begin{equation}
{\rm i}{\cal D}_{A'A}\lambda^A=\frac{\alpha}{\sqrt{2}}\bar\mu_{A'}, 
\hskip 20pt
{\rm i}{\cal D}_{AA'}\bar\mu^{A'}=\frac{\alpha}{\sqrt{2}}\lambda_A,
\label{eq:4.1}
\end{equation}
of equations. Taking the action of ${\cal D}^{BA'}$ on the first of these 
equations and eliminating $\bar\mu_{A'}$ by the second, we obtain that the 
eigenvalue problem is equivalent to $2{\cal D}^{AA'}{\cal D}_{A'B}\lambda^B=
\alpha^2\lambda^A$. Moreover, (\ref{eq:4.1}) implies that the eigenvalue 
$\alpha$ is \emph{real}: $0\leq2\Vert{\cal D}_{A'A}\lambda^A\Vert^2_{L_2}=
2\langle\,{\cal D}^{AA'}{\cal D}_{A'B}\lambda^B,\lambda^A\,\rangle=\alpha^2
\Vert\lambda^A\Vert^2_{L_2}$. Now we show that the multiplicity of every 
non-zero eigenvalue $\alpha^2$ is even. In fact, if $\lambda^A$ is 
an eigenspinor, then $\mu^A=-{\rm i}\frac{\sqrt{2}}{\alpha}{\cal D}^A{}_{A'}
\bar\lambda^{A'}$ is also an eigenspinor. But if $\mu^A$ were not independent 
of $\lambda^A$, then $\mu^A=c\lambda^A$ would hold for some non-zero complex 
constant $c$. Substituting this back into (\ref{eq:4.1}) we obtain that 
$\alpha(1+\vert c\vert^2)\lambda^A=0$, which would contradict $\alpha\not=0$. 
Thus, for each Dirac eigenspinor $\Psi^\alpha$, we have a pair of independent 
eigenspinors of $2{\cal D}^{AA'}{\cal D}_{A'B}$. 

Applying the basic norm identity to the eigenspinor $\lambda^A$ we obtain that 
$\alpha^2\geq\frac{3}{4}\kappa\,{\tt M}$, i.e. $\frac{3}{4}\kappa\,{\tt M}$ is 
a \emph{lower bound} for all the eigenvalues of $2{\cal D}^{AA'}{\cal D}_{A'B}$. 
We will see that this is precisely the smallest eigenvalue $\alpha^2_1$, i.e. 
a \emph{sharp} lower bound. The proof of this statement is based on the 
following functional analytic properties of the operator: 

\begin{theorem}
There is a dense subspace ${\rm Dom}({\cal D}^*{\cal D})\subset H_1(\Sigma,
\mathbb{S}^A)$ of the first Sobolev space of the unprimed spinor fields on 
$\Sigma$ such that it contains the space $C^\infty(\Sigma,\mathbb{S}^A)$ of 
the smooth spinor fields, and there is an extension of the operator ${\cal 
D}^{AA'}{\cal D}_{A'B}$ from $C^\infty(\Sigma,\mathbb{S}^A)$ to ${\rm Dom}({\cal 
D}^*{\cal D})$ such that ${\cal D}^{AA'}{\cal D}_{A'B}:{\rm Dom}({\cal D}^*
{\cal D})\rightarrow L_2(\Sigma,\mathbb{S}^A)$ is a positive self-adjoint 
Fredholm operator with compact resolvent. \label{th:4.1}
\end{theorem}
There is a similar result for the composition ${\cal T}^*{\cal T}$ of the 
3-surface twistor operator and its adjoint, too. The detailed proof of these 
statements is given in the Appendix of \cite{Sz12}. 

The significance of this theorem is that, via standard theorems of functional 
analysis, it guarantees that (1) the spectrum of $2{\cal D}^{AA'}{\cal D}_{A'B}$ 
(and of $2{\cal T}^*{\cal T}$, too) is purely discrete, and that (2) the 
corresponding eigenspinors span the whole space $L_2(\Sigma,\mathbb{S}^A)$. 
Thus we can order the eigenvalues into the increasing sequence $\alpha^2_1
\leq\alpha^2_2\leq\cdots\leq\alpha^2_i\leq\cdots$, and the sequence of the 
corresponding independent eigenspinors are denoted by $\{\lambda^A_i\}$. 
(Because of the multiplicity of the eigenvalues we should allow equality in 
the sequence of the eigenvalues.) Clearly, the eigenspinors $\{\lambda^A_i\}$ 
can be chosen to form an $L_2$-orthogonal system. 

Then let us expand the minimizer spinor field as $\lambda^A=\sum_ic_i\lambda
^A_i$, where $c_i\in\mathbb{C}$. Substituting this form of $\lambda^A$ into 
(\ref{eq:3.2.2}) we find 

\begin{equation}
0=\langle\,2{\cal D}^{AA'}{\cal D}_{A'B}\lambda^B-\frac{3}{4}\kappa\,{\tt M}
\lambda^A\, ,\,\lambda^A\rangle=\sum_i\vert c_i\vert^2\big(\alpha^2_i-
\frac{3}{4}\kappa\,{\tt M}\bigr)\Vert\lambda^A_i\Vert_{L_2}. \label{eq:4.2}
\end{equation}
Taking into account that $\frac{3}{4}\kappa\,{\tt M}$ is a lower bound for 
\emph{all} the eigenvalues $\alpha^2$, we conclude that 

\begin{equation}
\alpha^2_1=\frac{3}{4}\kappa\,{\tt M}, \label{eq:4.3}
\end{equation}
otherwise $\lambda^A$ would have to be zero. Therefore, the total mass 
(density) of closed universes can be recovered as the first eigenvalue of the 
Sen--Witten operator, which result makes ${\tt M}$ a well \emph{computable} 
quantity. Looking at (\ref{eq:4.3}) from the point of view of the spectral 
characterization of geometries, it is an explicit formula for the first 
eigenvalue of the Sen--Witten operator rather than only a lower bound for it. 
Hence (\ref{eq:4.3}) is a generalization of the result \cite{HZ} of Hijazi 
and Zhang.


\section{On the gauge conditions}
\label{sec:5}

An immediate consequence of the basic norm identity and Theorem \ref{th:3.1} 
is that \emph{Witten's gauge condition admits non-trivial solution if and only 
if ${\tt M}=0$}, i.e. precisely when the spacetime is holonomically trivial 
with toroidal $\Sigma$. 

To discuss the existence of solutions to the approximate twistor equation we 
need the explicit form of ${\cal T}^*{\cal T}$. First, the formal adjoint of 
the 3-surface twistor operator, $\mu_{ABC}\mapsto{\cal T}^*(\mu)_A$, $\mu_{ABC}
=\mu_{(ABC)}$, is 

\begin{equation}
{\cal T}^*\bigl(\mu\bigr)_A={}^+{\cal D}^{BC}\mu_{ABC}, \label{eq:5.1}
\end{equation}
where ${}^+{\cal D}_{AB}\lambda_C:=D_{AB}\lambda_C-\frac{1}{\sqrt{2}}\chi_{ABC}
{}^D\lambda_D$ and $\chi_{ABCD}$ is the unitary spinor form of the extrinsic 
curvature. Then, comparing the explicit form of ${\cal T}^*{\cal T}$ with 
that of ${\cal D}^{AA'}{\cal D}_{A'B}$ (given e.g. in terms of the Levi-Civita 
derivative operator $D_a$ and the extrinsic curvature $\chi_{ab}$), we find 
that 

\begin{equation}
{\cal T}^*{\cal T}(\lambda)_A=\frac{4}{3}{\cal D}_{AA'}{\cal D}^{A'B}
\lambda_B-\kappa\,t_eT^{eA'B}t_{A'A}\lambda_B. \label{eq:5.2}
\end{equation}
Thus, apart from a zeroth order operator, the approximate twistor operator is 
essentially the square of the Sen--Witten operator. Since the kernel of the 
Sen--Witten and the square of the Sen--Witten operators coincide, this 
implies that in vacuum it does \emph{not} admit any solution. 

On the other hand, the results of section \ref{sec:4} suggest a potentially 
viable alternative gauge condition. Namely, let us choose the 
\emph{eigenspinors} of the Sen--Witten operator corresponding to the 
\emph{first eigenvalue}. We showed in section \ref{sec:4} that there exist 
at least two such linearly independent spinor fields. Our conjecture is that 
the eigenspinors with the first eigenvalue can have no zeros, and hence, in 
particular, the number of these eigenspinors is precisely two. (In fact, if 
there were three such linearly independent eigenspinors, say $\lambda^A_1$, 
$\lambda^A_2$ and $\lambda^A_3$, then, since the spin space is two-complex 
dimensional, for any point $p\in\Sigma$ there would be non-zero complex 
constants $c_1$, $c_2$ and $c_3$ such that $c_1\lambda^A_1(p)+c_2\lambda^A_2(p)
+c_3\lambda^A_3(p)=0$ would hold, i.e. the eigenspinor $\lambda^A:=c_1\lambda
^A_1+c_2\lambda^A_2+c_3\lambda^A_3$ would have a zero at $p$.) Therefore, this 
gauge condition would yield a geometrically distinguished 3-parameter family 
of globally defined orthonormal \emph{vector bases} and \emph{lapse functions} 
on $\Sigma$.


\section{Examples}
\label{sec:6}

\subsection{Bianchi I. spacetimes}
\label{sec:6.1}

Let $\Sigma$ be a $t={\rm const}$ hypersurface in the Bianchi I. cosmological 
model with toroidal spatial topology. The induced intrinsic metric on $\Sigma$ 
is \emph{flat}, and let us write the corresponding line element as $dh^2=-(a^2d
\psi^2+b^2d\theta^2+c^2d\phi^2)$, where $a,b,c\in(0,\infty)$ and the coordinates 
are $\psi,\theta,\phi\in[0,2\pi)$. In the global orthonormal basis adapted to 
the spatial symmetries the extrinsic curvature of $\Sigma$ can be written as 
$\chi_{ab}={\rm diag}(\chi_1,\chi_2,\chi_3)$, where the diagonal elements are 
constant on $\Sigma$. 

First, we calculate the spectrum of three differential operators. The simplest 
one is the square of the Riemannian Dirac operator, built from the intrinsic 
(flat) Levi-Civita connection $D_e$. Since the Sen--Witten operator reduces 
to the Riemannian Dirac operator when the extrinsic curvature is vanishing, it 
seems natural to define its eigenvalue problem by $2D^{AA'}D_{A'B}\lambda^B=
\beta^2\lambda^A$. Then elementary calculations yield its spectrum: 

\begin{equation}
\beta=\pm\sqrt{(\frac{n_1}{a})^2+(\frac{n_2}{b})^2+(\frac{n_3}{c})^2}, 
\hskip 20pt n_1,n_2,n_3\in\mathbb{Z}, \label{eq:6.1.1}
\end{equation}
Thus, from the first few eigenvalues of the Riemannian Dirac operator we can 
recover the constants $a$, $b$ and $c$, i.e. \emph{the spatial geometry 
$(\Sigma,h_{ab})$ can be characterized completely by the spectrum of the 
Riemannian Dirac operator}. 

Since $\chi_{ab}$ is constant on $\Sigma$, it is easy to see that $2{\cal D}
^{AA'}{\cal D}_{A'B}\lambda^B=2D^{AA'}D_{A'B}\lambda^B+\frac{1}{4}\chi^2\lambda^A$. 
Therefore, if the eigenvalue problem for the Sen--Witten operator is defined 
by $2{\cal D}^{AA'}{\cal D}_{A'B}\lambda^B=\alpha^2\lambda^A$ (as in section 
\ref{sec:4}), then 

\begin{equation}
\alpha^2=\beta^2+\frac{1}{4}\chi^2. \label{eq:6.1.2}
\end{equation}
Thus \emph{the spatial geometry and the mean extrinsic curvature can be 
characterized completely by the spectrum of the Sen--Witten operator}. 

Finally, we define the eigenvalues of the 3-surface twistor operator by 
$2{\cal T}^*{\cal T}(\lambda)_A=\tau^2\lambda_A$, where the adjoint 
${\cal T}^*$ of the 3-surface twistor operator has been given explicitly by 
(\ref{eq:5.1}). For its eigenvalues we obtain that 

\begin{equation}
\tau^2=\frac{4}{3}\beta^2+\frac{1}{2}\big(\chi_{ab}-\frac{1}{3}\chi h_{ab}\big)
\big(\chi^{ab}-\frac{1}{3}\chi h^{ab}\big). \label{eq:6.1.3}
\end{equation}
Hence \emph{the spatial geometry and the magnitude of the trace free part of 
$\chi_{ab}$ can be characterized completely by the spectrum of the 3-surface 
twistor operator}. It could be interesting to see whether, in addition to the 
trace of $\chi_{ab}$ and of $\chi_{ac}\chi^c{}_b$, the trace of the \emph{cube} 
of the extrinsic curvature also can be recovered from the spectrum of some 
additional (probably higher order) elliptic operator. In this case we would 
have a \emph{complete} characterization of the initial data sets for the 
geometry of the closed Bianchi I. cosmological spacetimes. 

Next, let us calculate the total mass density in the closed Bianchi I. 
cosmological model. Since by the Hamiltonian constraint $\chi^2=2\kappa\mu+
\chi_{ab}\chi^{ab}$, by (\ref{eq:4.3}) and (\ref{eq:6.1.2}) it is 

\begin{equation}
{\tt M}=\frac{1}{3\kappa}\chi^2=\mu+\frac{1}{2\kappa}(\chi_{ab}-\frac{1}{3}
\chi h_{ab})(\chi^{ab}-\frac{1}{3}\chi h^{ab}). \label{eq:6.1.4}
\end{equation}
Thus the anisotropy of the extrinsic curvature, which is essentially the first 
eigenvalue of the 3-surface twistor operator, contributes to ${\tt M}$. 
Equation (\ref{eq:6.1.4}) illustrates how Theorem \ref{th:3.1} works: ${\tt M}
=0$, together with the Hamiltonian constraint, really imply flatness. 

We calculate the time derivative of ${\tt M}$ with respect to an evolution 
vector field compatible with the spacetime symmetries, $K^a=Nt^a$, where $N$ 
is constant on $\Sigma$. It is 

\begin{equation}
\dot{\tt M}=N\chi\Bigl(\frac{1}{3}h^{ab}\sigma_{ab}-\mu-\frac{1}{\kappa}\bigl(
\chi_{ab}-\frac{1}{3}\chi h_{ab}\bigr)\bigl(\chi^{ab}-\frac{1}{3}\chi h^{ab}\bigr)
\Bigr), \label{eq:6.1.5}
\end{equation}
where $\sigma_{ab}:=P^c_aP^d_bT_{cd}$, the spatial stress of the matter fields, 
and whose trace gives the average (isotropic) pressure: $p=-\frac{1}{3}h^{ab}
\sigma_{ab}$. In the `mean expanding phase' (i.e. when $\chi>0$) with `normal' 
matter (i.e. $\mu,p\geq0$) the total mass density ${\tt M}$ is 
\emph{decreasing}. Similarly, we can compute the time derivative of the total 
mass ${\tt M}{\rm vol}(\Sigma)$, too. In the `mean expanding phase' it is also 
negative. This behaviour is compatible with the interpretation that ${\tt M}$ 
is a positive definite measure of the strength of the gravitational field.


\subsection{FRW spacetimes}
\label{sec:6.2}

In the initial data set for a closed Friedman--Robertson--Walker spacetime 
there are only two independent geometrical quantities, the spatial scalar 
curvature ${\cal R}$ and the trace $\chi$ of the extrinsic curvature. These 
can be recovered from the first eigenvalue of the Riemannian Dirac and of the 
Sen--Witten operator, respectively: $\beta^2_1=\frac{3}{8}{\cal R}$, $\alpha
^2_1=\beta^2_1+\frac{1}{4}\chi^2$. 

To calculate the total mass density we need the Hamiltonian constraint. It 
is $\frac{1}{2}{\cal R}+\frac{1}{3}\chi^2=\kappa\mu$. Thus 

\begin{equation}
{\tt M}=\mu, \hskip 20pt
\dot{\tt M}=\frac{1}{3}N\chi\Bigl(h^{ab}\sigma_{ab}-3\mu\Bigr).
\label{eq:6.2.1}
\end{equation}
In the expanding phase (i.e. when $\chi>0$) with `normal' matter both ${\tt 
M}$ and ${\tt M}{\rm vol}(\Sigma)$ are \emph{decreasing}.


\section{Summary}
\label{sec:7}

The quantity ${\tt M}$, defined by equation (\ref{eq:3.1.2}) with the set of 
smooth spinor fields satisfying appropriate boundary conditions at infinity 
on asymptotically flat or asymptotically hyperboloidal hypersurfaces, provides 
a positive lower bound for the ADM and Bondi--Sachs masses, respectively. 

On \emph{closed} hypersurfaces for the \emph{same} ${\tt M}$ (defined with a 
different set of the spinor fields) the following statements have been proven: 

\begin{itemize}
\item ${\tt M}=0$ iff the spacetime is holonomically trivial with toroidal 
Cauchy hypersurface,

\item ${\tt M}$ gives the first eigenvalue $\alpha^2_1$ of the square 
$2{\cal D}^{AA'}{\cal D}_{A'B}$ of the Sen--Witten operator: 

$$
\alpha^2_1=\frac{3}{4}\kappa\,{\tt M}, \nonumber
$$

\item Witten's gauge condition, ${\cal D}_{A'A}\lambda^A=0$, admits a 
non-trivial solution iff ${\tt M}=0$. 

\end{itemize}
Here we showed that in general the so-called approximate twistor operator 
cannot be used to determine a gauge condition in closed universes. 
Nevertheless, we suggested an alternative gauge condition, viz. the use of 
the \emph{eigenspinors} of $2{\cal D}^{AA'}{\cal D}_{A'B}$ corresponding to 
the first eigenvalue above.

Through simple examples we illustrated how the geometry of the data sets 
for closed universes could be characterized by the spectrum of the Sen--Witten 
and the 3-surface twistor operators. In these examples we also calculated 
the quantity ${\tt M}$ and its time derivative. The results support the 
interpretation of ${\tt M}$, suggested the general properties listed above: 
It, as a positive definite measure of the strength of the gravitational 
`field', can be interpreted as the \emph{total mass density} of closed 
universes at the instant represented by the closed hypersurface $\Sigma$. 
Nevertheless, the ultimate answer to the question whether this is a 
reasonable and useful notion will be given by the future applications. 

\medskip
I would like to thank the organizers for the invitation to the Spanish 
Relativity Meeting in Portugal at Guimar\~{a}es, 2012 September, where these 
results could be presented.

\end{document}